\documentclass{elsarticle}

\usepackage{amsmath}
\usepackage{amssymb}
\usepackage{dcolumn}
\usepackage{amsthm}
\usepackage{amsfonts}
\usepackage{graphicx}
\usepackage{dcolumn}
\usepackage{stackrel}

\newtheorem{thm}{Theorem}
\newtheorem*{defin}{Definition}
\begin{document}
\begin{frontmatter}
\title{CAR T cells for T-cell leukemias: \\ Insights from mathematical models}

\author{V\'{\i}ctor M. P\'erez-Garc\'{\i}a}
\address{Mathematical Oncology Laboratory (MOLAB), Departamento de
Matem\'aticas, E. T. S. I.  Industriales and Instituto de Matem\'atica Aplicada a la Ciencia y la Ingenier\'{\i}a, Universidad de Castilla-La
Mancha, 13071 Ciudad Real, Spain.
\texttt{victor.perezgarcia@uclm.es} }

\author{Odelaisy Le\'on-Triana}
\address{Mathematical Oncology Laboratory (MOLAB), Departamento de
Matem\'aticas, E. T. S. I.  Industriales and Instituto de Matem\'atica Aplicada a la Ciencia y la Ingenier\'{\i}a, Universidad de Castilla-La
Mancha, 13071 Ciudad Real, Spain.
\texttt{odelaisy.leon@uclm.es} }

\author{Mar\'{\i}a Rosa}
\address{Department of Mathematics, Universidad de C\'{a}diz, Puerto Real, C\'{a}diz, Spain. \texttt{maria.rosa@uca.es}}

\author{Antonio P\'erez-Mart\'{\i}nez}
\address{Translational Research Unit in Paediatric Haemato-Oncology, Hematopoietic Stem Cell Transplantation and Cell Therapy, Hospital Universitario La Paz, Madrid, Spain and 
Paediatric Haemato-Oncology Department, Hospital Universitario La Paz, Madrid, Spain}

\date{\today}

\begin{abstract}
Immunotherapy has the potential to change the way all cancer types are treated and cured. Cancer immunotherapies use elements of the patient immune system to attack tumor cells. One of  the most successful types of immunotherapy is CAR-T cells. This treatment works by extracting patient’s T-cells and adding to them an antigen receptor allowing tumor cells to be recognized and targeted. These new cells are called CAR-T cells and are re-infused back into the patient after expansion in-vitro. 
This approach has been successfully used to treat B-cell malignancies (B-cell leukemias and lymphomas). However, its application to the treatment of T-cell leukemias faces several problems. One of these is fratricide, since the CAR-T cells target both tumor and other CAR-T cells. This leads to nonlinear dynamical phenomena amenable to mathematical modeling. 

In this paper we construct a mathematical model describing the competition of CAR-T, tumor and normal T-cells and studied some basic properties of the model and its practical implications. Specifically, we found that the model reproduced the observed difficulties for in-vitro expansion of the therapeutic cells found in the laboratory. The mathematical model predicted that CAR-T cell expansion in the patient would be possible due to the initial presence of a large number of targets. We also show that, in the context of our mathematical approach, CAR-T cells could control tumor growth but not eradicate the disease.

\end{abstract}

\begin{keyword}
Mathematical Oncology, T-cell leukemias, CAR-T cell therapies
\end{keyword}
\end{frontmatter}

\section{Introduction}

\label{intro}

Cancer immunotherapies use elements of the patient immune system to attack tumor cells. These treatments encompass different therapeutic strategies typically involving collecting a specific set of cells from patients, modifying them to produce some kind of attack on cancer cells, and reinjecting them. Some examples are tumor-infiltrating lymphocytes, engineered T-cell receptor, chimeric antigen receptor (CAR)-T cells, cytotoxic T-lymphocytes, natural killer cells, and mesenchymal stem cells \cite{Rafei2019}.

Of these, the most successful type of immunotherapy today is CAR-T cells. This treatment works by extracting patient’s T-cells and adding the CAR group to them, allowing them to recognize and target the cells carrying an antigen expressed in the tumor \cite{Feins2019}. 
The case of B-cell leukemias expressing CD19 has been particularly successful since this antigen is only expressed by B-lymphocytes and B-lymphoid leukemia cells. The clinical use of CAR-T cells engineered to recognize this antigen have led to
the full recovery of a large fraction of Acute Lymphoblastic Leukemia patients \cite{Feins2019,ALL1,Pan2017,Militou}. Good results have been reported for large B-cell lymphomas  \cite{Lymphoma1,Lymphoma2} and multiple myelomas \cite{MMyeloma}. These successes have led to the approval of CAR-T therapies directed against CD19 for treatment of acute lymphoblastic leukemias and diffuse large B-cell lymphomas \cite{Cell}. However, CAR-T cell therapies have not yet been as successful for solid tumors, for a variety of different reasons \cite{Yong2017,Moon2019}. 

Mathematical modeling has the potential to help in finding optimal administration protocols, provide a deeper understanding of the dynamics, help in the design of clinical trials and more. 
The clinical relevance of CAR-T cells has attracted the attention of applied mathematicians that have started to construct mathematical models and study different aspects of these therapies \cite{Sahoo,Baar,Kimmel,Rodrigues,Anna,Ode}. 

Given the success of CAR-T cells directed against CD19 in B-cell malignancies, new targets are being developed and tested. Specifically, there has been a lot of interest in the possibility of using CAR-T cells for the treatment of T-cell malignancies \cite{Alcantara2019,Breman2018,Fleischer2019,Menendez2019}. However there are many challenges in translating this therapy for T-cell disease. The first one is fratricide, which refers to the mutual killing of CAR T-cells. This phenomenon may prevent the generation, expansion and persistence of CAR-T cells. The second one is the prolonged and profound T-cell aplasia induced by the destruction of normal T-cells, that exposes patients to severe opportunistic infections. 
The third one is the potential contamination of CAR T-cell products with malignant T-cells. Indeed, circulating tumor T-cells are often found in the peripheral blood of patients. Because tumor T-cells may harbor the same properties as normal T-cells, they may be harvested, transduced, expanded, and infused concomitantly with normal T-cells as described recently in the context of B-cell leukemias \cite{NatureMedicine2018}. 
 Thus, developing CAR-T cells for T-cell malignancies requires avoiding contamination of the CAR-T cell product with malignant transduced T-cells \cite{Alcantara2019}. 

To the best of our knowledge, no mathematical model has yet considered CAR-T cell treatments for T-cell malignancies.

 In this paper we want to build the first minimal mathematical model describing the dynamics of tumor cells in T-cell leukemias and normal T-cells plus a population of injected CAR-T cells. Our intention is to describe the effect of the fratricide mathematically and to obtain conclusions of practical interest. This interesting phenomenon, which involves a nonlinear self-interaction within the CAR-T cell compartment will be shown to place a limit on the production of these cells in vitro. Our theoretical and simulation results support that CAR T-cells could be able to control tumor growth in vivo to a certain extent. We will show that it may not be possible to get rid of all tumor cells, but that the treatment could be useful either as a bridge treatment or as a way of making the disease chronic.

Our focus in this paper was to perform a preliminary exploration of the biological problem and obtain conclusions of practical applicability, using numerical simulations of the mathematical model as a test bed.

The structure of the paper is as follows. First, in Sec. \ref{model} we set out the mathematical model, estimate its parameters and perform a basic study of some of its properties and study the model's equilibria. Next, in Sec. \ref{results} we consider different scenarios including the generation of the CAR-T product in vitro and the in-vivo dynamics. Finally, Sec. \ref{conclu} discusses our findings and summarizes our conclusions. 

\section{The Model}\label{model}
\par

\subsection{Basic Mathematical Model}\label{Derivation}
\par

 Our mathematical model accounts for the dynamics of several cell populations: CAR-T cells $C(t)$, leukemic T-cells $L(t)$, and normal T-cells $T(t)$.  The equations describing the dynamics of these populations are
 \begin{subequations}
 \label{model1}
 \begin{eqnarray} \label{model11}
 \frac{dC}{dt} & = & \rho_C \left( T + L + C\right) C - \frac{1}{\tau_C} C - \alpha C^2 + \rho_I C, \\ \label{model12}
 \frac{dL}{dt} & = & \rho_L L - \alpha L C, \\ \label{model13}
 \frac{dT}{dt} & = & g(T,L,C)  - \alpha  T C.
 \end{eqnarray}
 \end{subequations}

  CAR-T cells,  described by Eq. (\ref{model11}), have a finite lifespan $\tau_C$ and proliferate due to stimulation by target cells (either $L(t)$ or  $T(t)$ or the CAR-T cell themselves $C(t)$). The parameter $\rho_C$ measures the stimulation of mitosis after encounters with target cells. The parameter $\alpha$ in Eq. (\ref{model11}) is a cell kill term accounting for the fratricide. It measures the probability that CAR-T cell encounters lead to the death of one of the cells. Once the CAR-T cell identifies the target cell, killing and detachment are very fast processes \cite{KillingCAR}. We consider here only serial killing, excluding multiplexed killing, which would be a less relevant process and have a different kinetics.
  
  In line with models for CAR-T cell dynamics in B-cell leukemias \cite{Ode}, we did not include a CAR-T cell death term due to encounters with target cells. The reason is that CAR-T cells do not die after killing target cells \cite{SerialKillers1,SerialKillers2}. Also, T-cells do not divide in vivo spontaneously \cite{Tough1995},  their clonal expansion being dependent on the stimulation with the target antigen, thus in vivo $\rho_I = 0$. When CAR-T cells are expanded in-vitro cytokines are added externally forcing the cells to divide, thus in that context we will assume $\rho_I \neq 0$.

Leukemic cells [Eq. (\ref{model12})] proliferate with a rate $\rho_L$ and die to the encounters with the CAR-T cells with the rate $\alpha$. 

For the normal T-cell compartment we will only consider a simplified effective description accounting for the different lineages expressing the same target antigen in an aggregate form. These cells will be assumed to be killed at a rate $\alpha$ per cell assumed to be similar to that of the other subpopulations and will be produced at a rate $g(T,L,C)$. This function is expected to depend on the total number of T-cells via cytokine signaling, on the effect of CAR-T cell on T-cell progenitors, etc. In this paper we will assume $g(T,L,C)$ to be very small and contribute only to a minimal residual level of normal T-cells that would not be relevant for the nonlinear dynamics of the system. In what follows we will take $g(T,L,C) = 0$.

Figure \ref{fig1} summarizes  the relationships between the different cell subpopulations and the assumptions behind our model.

\begin{figure}
\centering
	\includegraphics[width=\textwidth]{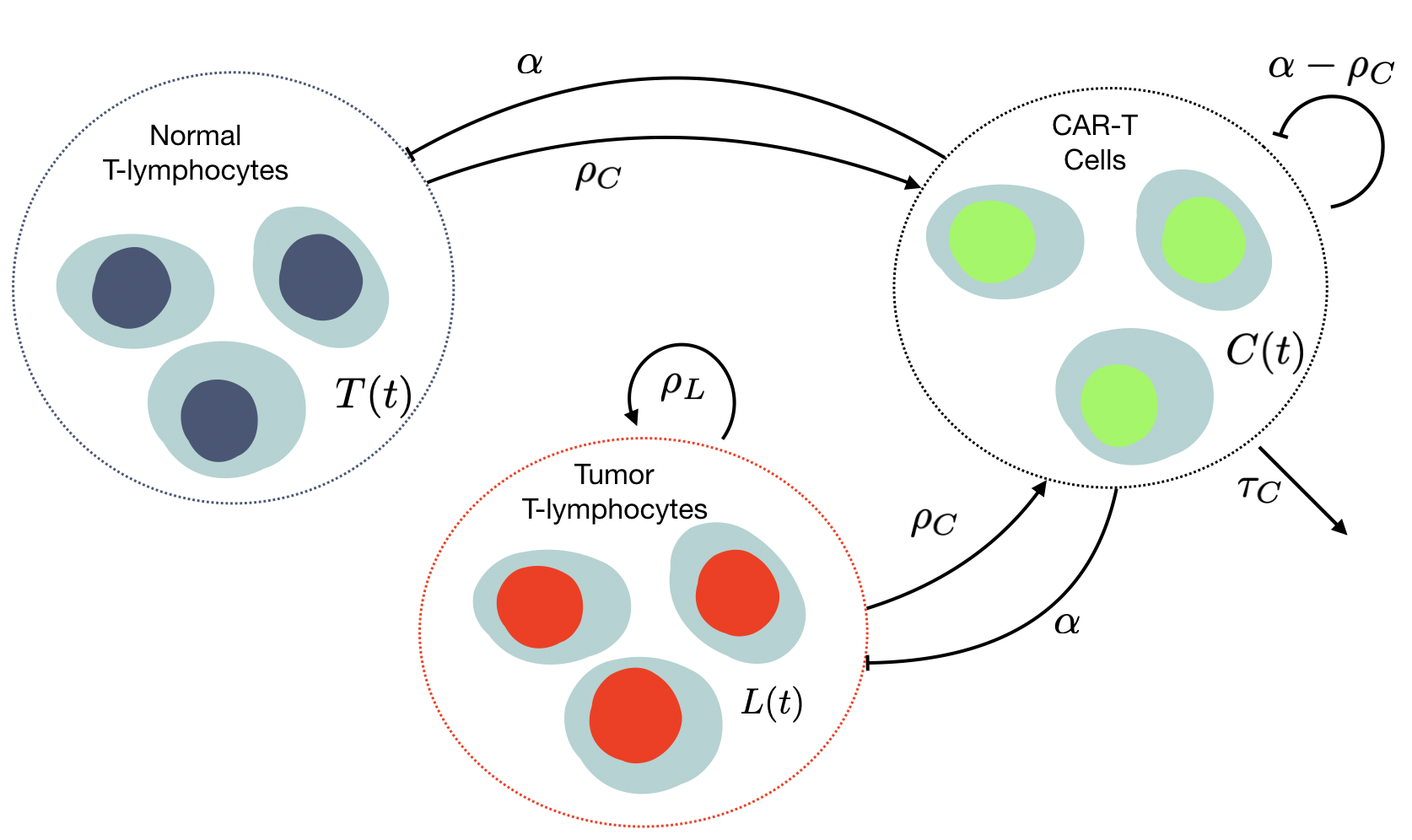}
	\caption{\textbf{Cellular populations and biological processes included in the mathematical model (\ref{model1})}. Normal $T(t)$ and leukemic $L(t)$ T-lymphocytes are killed by CAR-T cells $C(t)$ at a rate $\alpha$ and stimulate CAR-T cell proliferation at a rate $\rho_C$, both per cell. Tumor lymphocytes proliferate at a rate $\rho_L$. As a result of fratricide and self-stimulation, CAR-T cells are eliminated at a rate $\alpha-\rho_C$ per CAR-T cell. CAR-T cell finite lifetime $\tau_C$ also results in cell loss.}
\label{fig1}
\end{figure}

 \subsection{In-vitro equilibria}
 
The in-vitro expansion during the CAR-T cell production can be described by setting $L=T=0$ in Eqs. (\ref{model1}), and taking $\rho_I \neq 0$, thus
\begin{equation}\label{model2}
 \frac{dC}{dt}  = \hat{\rho}  C - \hat{\alpha} C^2.
 \end{equation}
The parameter $\hat{\rho} = \rho_I- 1/\tau_C > 0$ always in vitro and the effective cell kill rate $\hat{\alpha} = \alpha - \rho_C \geq 0$,  since the kill rate is expected to be larger than the stimulation rate due to the different speeds of the killing and replication processes. 

Eq. (\ref{model2}) is a logistic equation, that for positive initial values satisfies that 
\begin{equation}\label{expansion}
C \stackrel[t \rightarrow \infty]{}{\longrightarrow} C_* = \hat{\rho}/\hat{\alpha}.
\end{equation}
 This result is in line with the observation that CAR-T cells targeting T-cell antigens cannot be expanded beyond a certain value \cite{Breman2018}. Here we show that this value will depend on the cytokine stimulation provided and the fratricidal cell killing rate.

\subsection{Parameter estimation}
\label{parapara}

Some parameters in model Eq. (\ref{model1}) can be estimated from biological data. Firstly, the typical lifetime of activated CAR-T cells $\tau_c$  is in the range 14-28 days \cite{Ghorashian}. T-cell leukemias are typically rather aggressive tumors with small doubling times that can be estimated to be around $\rho_L$ = 1/40 day$^{-1}$ in vivo \cite{Doubling}, although chronic forms of the disease could have much smaller numbers \cite{PLL}. Finally, $\alpha$ and $\rho_C$ can be expected to be in the range of B-cell leukemias, where they have been found to be around $10^{-11}$ day$^{-1}$cell$^{-1}$ \cite{Ode}. One would expect $\rho_C$ to be of the order of or smaller than $\alpha$, since it corresponds to the number of new cells generated by each encounter of CAR-T cells with target cells.

As to the initial data, the total number of T-lymphocytes in the human body is around 10$^{11}$ and  typical tumor loads in acute T-cell leukemias can be in a similar range \cite{Bains2009}. Most CAR-T administration regimes are preceded by a lympho-depleting treatment that creates a favorable 
cytokine profile, favoring the growth of injected cells \cite{Depletion1,Depletion2}. Thus, the previous numbers are substantially reduced once the treatment is started. We will take our initial data to be around $\sim 10^{10}$ for tumor and normal T-cells.

Finally, the number of CAR-T cells injected would depend on the maximal expansion obtained in vitro, which could range from as low as $10^4$ when fratricide is present to larger numbers around $10^7$ depending on the strategies used to overcome it.

\subsection{Large initial data display unbounded dynamics}
\label{uns}

\begin{thm}
For any non-negative initial data $(C_0,L_0,T_0)$ and all the parameters of the model being positive, the solutions to Eqs. (\ref{model1}) exist for $t > 0,$ are non-negative and unique.
\end{thm}

\begin{proof} The ODE system (\ref{model1}) has bounded coefficients and the right-hand side of the system is a continuous function of $(C,L,T)$, thus the local existence of solutions follows from classical ODE theory. Since the partial derivatives of the velocity field are continuous and bounded, uniqueness follows from the Picard-Lindelof theorem. 

Let us rewrite Eqs. (\ref{model1}) when $g = 0$ as
 \begin{subequations}
 \label{modelx1}
 \begin{eqnarray} \label{modelx11}
 \dot{C} & = & \left[ \rho_C \left( T + L + C\right) -1/\tau_C - \alpha C\right] C, \\ \label{modelx12}
 \dot{L} & = & \left(\rho_L-\alpha C\right) L, \\
 \dot{T} & = & - \left( \alpha  C\right)T ,
 \end{eqnarray}
 \end{subequations}
then we may write 
 \begin{subequations}
 \begin{eqnarray} 
C(t)  & = & C_0 \exp \left(\int_{t_0}^t \left[ \rho_C T(t') + \rho_C L(t') + (\rho_C-\alpha)C(t') -\frac{1}{\tau_C}  \right] dt' \right) , \\
L(t) & = & L(t_0)  \exp \left(\int_{t_0}^t \left(\rho_L-\alpha C(t')\right) dt'\right), \\
T(t) & = &  T(t_0) \exp \left(- \int_{t_0}^t \alpha  C(t') dt' \right),
 \end{eqnarray}
 \end{subequations}
which leads to the positivity of solutions.

\end{proof}

\begin{defin} The sum of all cell populations studied will be denoted by $S(t)$, i.e.
$$S(t) = C(t) + T(t) + L(t).$$
\end{defin}

\begin{thm} Let $C(t), L(t), T(t)$ be solutions of Eqs. (\ref{model1}) with initial data $C(t_0) = C_0 > 0, L(t_0) = L_0 > 0, T(t_0)= T_0 > 0, S(t_0) = S_0>0$. If
\begin{description}
\item[H1] $\rho_C >\alpha$
\item[H2] $(\rho_C-\alpha) \tau_C S_0 > 1$,
\end{description}
 then $S(t)$ increases monotonically with time and 
$\lim_{t\rightarrow \infty} S(t) = \infty$.
\end{thm}

\begin{proof} Let us first sum the three equations Eq. (\ref{model1}) to obtain
$$ \frac{dS}{dt} = \left(\rho_C-\alpha\right) SC - \frac{1}{\tau}C + \rho_L L. $$
Then, the positivity of $L(t)$ implies that, $\underline{S}$ defined as the solution of 
\begin{equation}\label{derivi}
 \frac{d\underline{S}}{dt} = \left(\rho_C-\alpha\right) \underline{S}C - \frac{1}{\tau}C,
 \end{equation}
satisfying $\underline{S}(t_0) = S_0$ with $C(t_0) = C_0$, is a subsolution of $S(t)$, i.e. satisfying $\underline{S}(t) < S(t), \forall t>t_0$. Clearly, under our hypothesis
$$ \left.\frac{d\underline{S}}{dt}\right|_{t=t_0} = \left[\left(\rho_C-\alpha\right) \underline{S}_0 - \frac{1}{\tau}\right]C_0 > 0,$$
but then, using Eqs. (\ref{derivi}), this leads to $d\underline{S}/dt > 0$ for all $t>t_0$. Moreover, from Eq. (\ref{model11}) and using the fact that $S(t) > C(t)$, for all $t>t_0$ we get
\begin{equation}
 \frac{dC}{dt}  =  \rho_C S C - \frac{1}{\tau_C} C - \alpha C^2 >  
 \left[ \left( \rho_C-\alpha\right) S - \frac{1}{\tau_C}\right] C,
 \end{equation}
 where we have used $\rho_C S C - \frac{1}{\tau_C} C - \alpha C^2 > \rho_C S C - \frac{1}{\tau_C} C - \alpha C\left(C+L+T\right)$.

This means that $C(t) > C_0$ for any non-zero initial data, then 
\begin{equation}\label{deriva}
 \frac{d\underline{S}}{dt} > \left(\rho_C-\alpha\right) \underline{S}_0 - \frac{1}{\tau} C_0 \equiv Q_0 > 0,
 \end{equation}
then $S(t) > \underline{S}(t) > Q_0 t + S_0$, which proves the unboundedness of the total population $S(t)$, i.e. the fact that $\lim_{t\rightarrow \infty} S(t) = \infty$.
\end{proof}

Taking reasonable initial numbers (see Sec. \ref{parapara}) we will always initially be in the regime $(\rho_C-\alpha) \tau_C S_0 > 1$.

\subsection{Equilibria of the model Eqs. (\ref{model1}) and local stability analysis}
 
The equilibria of Eqs. (\ref{model1}) in the case of interest $g =0$, are given by the equations 
\begin{subequations}
  \begin{eqnarray} \label{equil11}
0 & = & \rho_C \left( T + L + C\right) C - \frac{1}{\tau_C} C - \alpha C^2, \\ \label{equil12}
0 & = & \rho_L L - \alpha L C, \\ \label{equil13}
0& = &  - \alpha  T C.
 \end{eqnarray}
 \end{subequations}
 Eq. (\ref{equil13}) leads to either $T = 0$ or $C=0$. The latter leads to $L=0$ using Eq. (\ref{equil12}) and the former to either $L=0$ or $C=\rho_L/\alpha$. Then using Eq. (\ref{equil11}) allows us to obtain the expressions for the three equilibrium points of Eqs. (\ref{model1})
 \begin{subequations}\label{equi}
 \begin{eqnarray}
 E_1 & = & (0, 0, T_*),\\
 E_2 & = & \left( \frac{\rho_L}{\alpha}, \frac{1}{\rho_C\tau_C} + \frac{\rho_L}{\rho_C}-\frac{\rho_L}{\alpha}, 0 \right). \\
 E_3 & = & \left( \frac{1}{\tau_c \left(\rho_C-\alpha\right)}, 0, 0\right).
 \end{eqnarray}
 \end{subequations}
  for any $T_*$. The Jacobian of the differential equations (\ref{model1}) is
 \begin{equation}\label{jacobo}
 J = \left( \begin{array}{ccc} 2\left(\rho_C-\alpha\right) C - 1/\tau_C + \rho_C\left(T+L\right) & \rho_C C & \rho_CC \\
 -\alpha L & \rho_L - \alpha C & 0 \\ - \alpha T & 0 & - \alpha C \end{array} \right) .
 \end{equation}
Let us now use Eq. (\ref{jacobo}) to study the local stability of the different equilibria given by Eqs. (\ref{equi}). Firstly, for $E_1$ we get
  \begin{equation}
 J\left(E_1\right) = \left( \begin{array}{ccc} \rho_C T_*-1/\tau_C & 0 & 0 \\
 0 & \rho_L & 0 \\ - \alpha T_* & 0 & 0 \end{array} \right).
 \end{equation}
The eigenvalues of $J\left(E_1\right)$ are 
\begin{equation}
\lambda_1 = 0, \lambda_2 = \rho_L, \lambda_3 = \rho_C T_* - 1/\tau_C,
\end{equation} 
thus the equilibrium point $E_1$ is unstable. For the second equilibrium point we get  
 \begin{equation}
 J\left(E_2\right) = \left( \begin{array}{ccc} \rho_L\left( \frac{\rho_C}{\alpha}-1\right) & \rho_C \rho_L/\alpha & \rho_C \rho_L/\alpha \\
 -\alpha \left(\rho_L+1/\tau_C\right)/\rho_C + \rho_L & 0 & 0 \\ 0 & 0 & - \rho_L \end{array} \right).
 \end{equation}
 Thus $\lambda_3 = - \rho_L < 0$ and $\lambda_{1,2}$ satisfy the equation
 \begin{equation}
 \lambda^2 + \lambda \left(1  - \frac{\rho_C}{\alpha}\right) \rho_L  + \rho_L^2 \left(1- 
 \frac{\rho_C}{\alpha}\right) + \frac{\rho_L}{\tau} = 0,
 \end{equation}
which leads to the eigenvalues
\begin{equation}\label{equifor}
\lambda_{\pm} = \frac{1}{2} \rho_L \left( \frac{\rho_C}{\alpha} - 1\right) \pm \frac{1}{2} D^{1/2},
\end{equation}
with the discriminant $D$ being given by
\begin{equation}
\frac{D}{\rho_L^2} = \left(1-\frac{\rho_C}{\alpha}\right)^2 - 4 \left(1-\frac{\rho_C}{\alpha}\right) -  \frac{4}{\rho_L\tau_C}.
\end{equation}
Let us consider the case $\rho_C/\alpha <1$ since the system will be unstable otherwise (Sec. \ref{uns}). Then,
 $$ 0< 1-\frac{\rho_C}{\alpha}<1.$$
Since $(1-\rho_C/\alpha)^2 < 1-\rho_C/\alpha$, we get
 $$ \frac{D}{\rho_L^2} <  - 3 \left(1-\frac{\rho_C}{\alpha}\right) -  \frac{4}{\rho_L\tau_C} < 0.$$
 Thus, the equilibrium is a stable node-focus.
 
 Finally, for $E_3$ we get
 \begin{equation}
  J(E_3) = \frac{1}{\tau_C} \left( \begin{array}{ccc}  1 & \frac{\rho_C}{\rho_C-\alpha} & \frac{\rho_C}{\rho_C-\alpha} \\
 0 & \rho_L\tau_c + \frac{\alpha}{\alpha-\rho_C} & 0 \\ 0 & 0 & \frac{\alpha}{\alpha-\rho_C} \end{array} \right).
 \end{equation}
 The eigenvalues of $J(E_3)$ are 
 \begin{subequations}
 \begin{eqnarray}
 \lambda_1 & = & 1/\tau_C > 0, \\
 \lambda_2 & = & \rho_L\tau_c + \alpha/(\alpha-\rho_C) > 0,\\
  \lambda_3 & = & \alpha/(\alpha-\rho_C) > 0,
 \end{eqnarray}
 \end{subequations}
  thus $E_3$ is an unstable node.

In conclusion, there is only one stable equilibrium point $E_2$ given by Eqs. (\ref{equi}) of node-focus type, which can be an attractor for the dynamics of the system (\ref{model1}). 
\section{Applications.}
\label{results}

\subsection{CAR-T cells allow for control of T-cell leukemias in the presence of fratricide}

To obtain further insight into the global dynamics of solutions of Eqs. (\ref{model1}) we simulated different initial data in the biologically feasible parameter and initial data regions. In all cases studied, we found an oscillatory behavior of the solutions towards the stable node-focus point $E_2$ after a fast reduction of the initial normal T-cell number. 

Figure \ref{fig2} provides a typical example of the dynamics. There we see how tumor grows for a short time, typically 10-15 days, while CAR-T cells expand. The CAR-T cell expansion persists over more than four orders of magnitude in cell number (Figure \ref{fig2}(c)), with a peak at about 15 days after injection (Figure \ref{fig2}(b)). This leads to a substantial decrease of the tumor load and T-cell aplasia (Figure \ref{fig2}(a,b)). For this parameter set, tumor was not controlled for long periods of times and relapse was noticeable a few months after the injection date of CAR-T cells. 
After relapse oscillations of leukemic and CAR-T cells are observed in their course towards the equilibrium, in this case corresponding to $3\times 10^8$ CAR-T cells and $2.7\times 10^9$ tumor cells. Interestingly, the number of tumor cells in this case is one order of magnitude smaller than the initial tumor load ($2 \times 10^{10}$ cells), which supports the possibility of CAR-T cells effectively controlling tumor  to clinically acceptable levels.

The numerical results used to construct Figure \ref{fig2}(c), show that in less than two months after injection, treatment was able to reduce tumor load from the initial level of $2\times 10^{10}$ cells down to a minimum level of $2.59 \times 10^{8}$ cells, i.e. a decrease of about two orders of magnitude.

 \begin{figure}
\centering
	\includegraphics[width=0.65\textwidth]{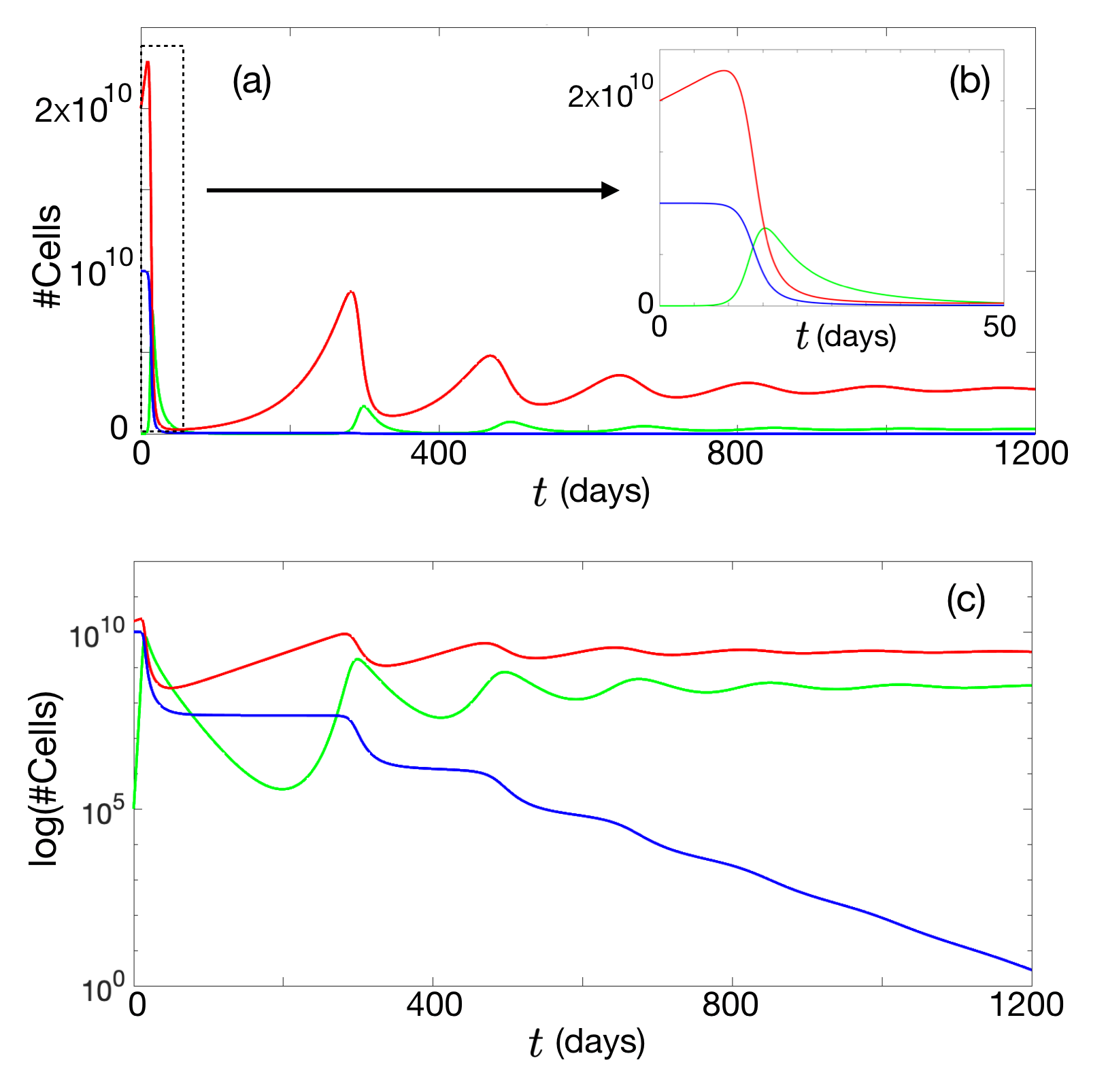}
	\includegraphics[width=0.65\textwidth]{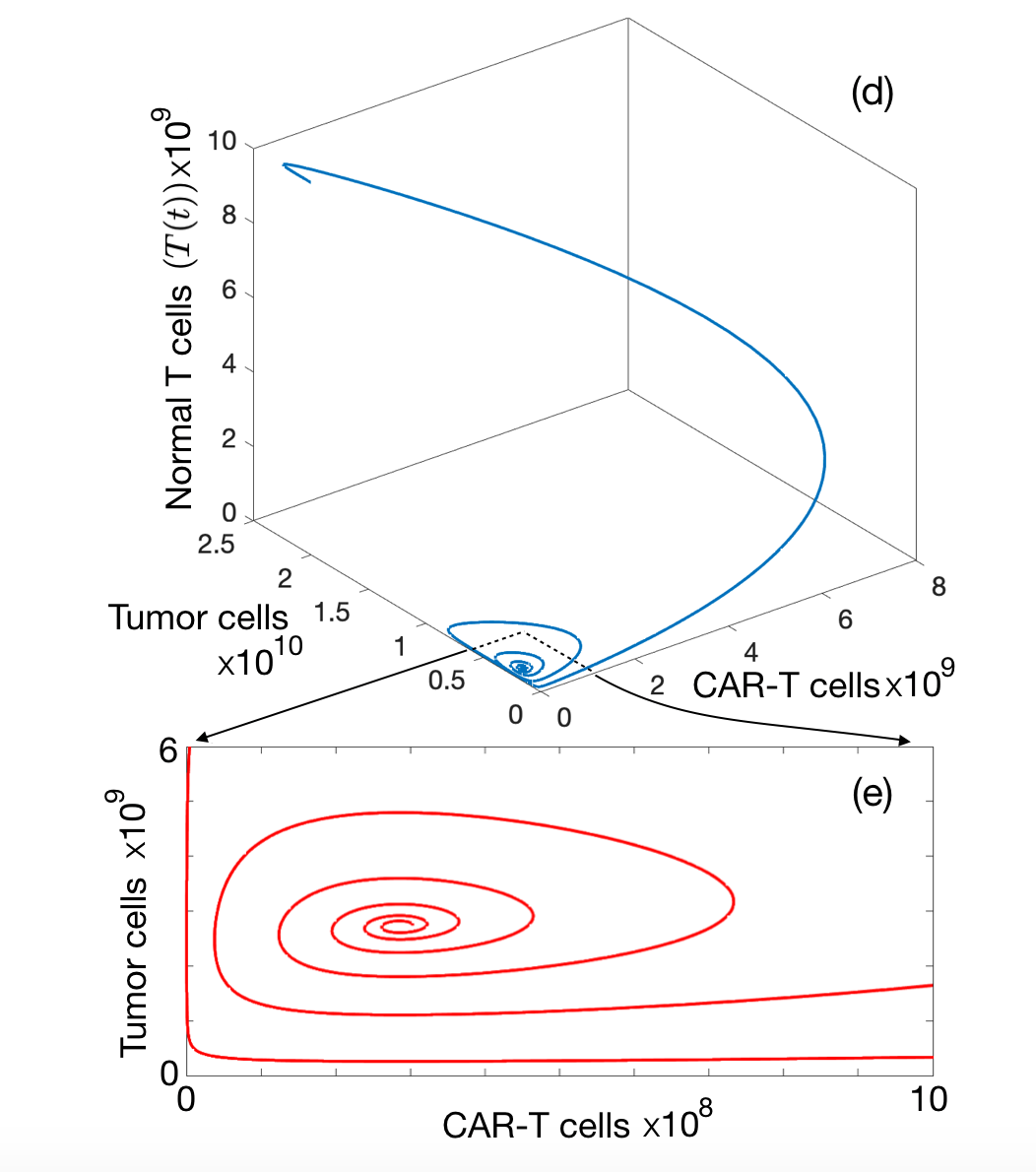}
	\caption{\textbf{Typical dynamics of cell populations governed by Eq. (\ref{model1})}. Results of a simulation are shown for parameter values $\tau_C = 14$ days, $\rho_L = 1/60$ day$^{-1}$, $\alpha = 5.84 \times 10^{-11}$ day$^{-1}$ cell$^{-1}$, $\rho_C = \alpha/2$, and initial data $T_0 = 10^9, L_0 = 2\times 10^9, C_0 = 10^5$ cells. (a-c) Dynamics of the populations of CAR-T (green), tumor (red) and normal T cells (blue). Dynamics are depicted on the time intervals  $t \in[0, 1200]$ (a,c) and $t \in[0, 50]$  (b)  , and in linear  (a,b) and logarithmic (c) scales. (d) Trajectory of the solution in the phase space. (e) Projection of the selected part of the trajectory on the $(T(t),C(t))$ plane.}
\label{fig2}
\end{figure}

\subsection{Higher mitotic stimulation rates improve tumor control}

The asymptotic equilibrium values of leukemic cells and CAR-T cells are given by $E_2$, i.e. 
\begin{subequations}
\begin{eqnarray}\label{L2e}
L_2 & = & \frac{1}{\rho_C} \left( \rho_L + \frac{1}{\tau_C}\right) - \frac{\rho_L\rho_C}{\alpha}, \\
C_2 & = & \frac{\rho_L}{\alpha}. 
\end{eqnarray}
\end{subequations}
Interestingly, the equilibrium level of CAR-T cells does not depend on the mitotic stimulation rate $\rho_C$, but only on the growth and death rates of leukemic cells. 
However, the most important thing, due to the clinical implications, are the leukemia equilibrium levels $L_2$, and the maximum leukemic load $\max_{t} L(t)$. Maxima would typically be attained in time
during the CAR-T cell expansion stage. 

Let us note that
\begin{equation}
\frac{dL_2(\rho_C)}{d\rho_C} =  -\frac{1}{\rho^2_C} \left( \rho_L + \frac{1}{\tau_C}\right) - \frac{\rho_L}{\alpha} < 0,
\end{equation}
this means that $L_2(\rho_C)$ is a monotonically decreasing function. Since $\rho_C \alpha>0$, the minimum of $L_2(\rho_C)$ over the range $\rho_C \in [0, \alpha]$ would be obtained when $\rho_C = \alpha$. Figure \ref{equilibriaL} 
confirms that the asymptotic values of $L_2$ decrease with the mitotic stimulation rate $\rho_C$ and thus larger values of the mitotic stimulation rate lead to better tumor control. However, going beyond $\rho_C = \alpha$ destabilizes the system, as discussed in Sec. \ref{uns}. Thus, it may be necessary to control in detail the CAR-T manufacturing process to get both high mitotic stimulation rates while at the same time not getting too close to the instability regime. 

 \begin{figure}
\centering
	\includegraphics[width=0.9\textwidth]{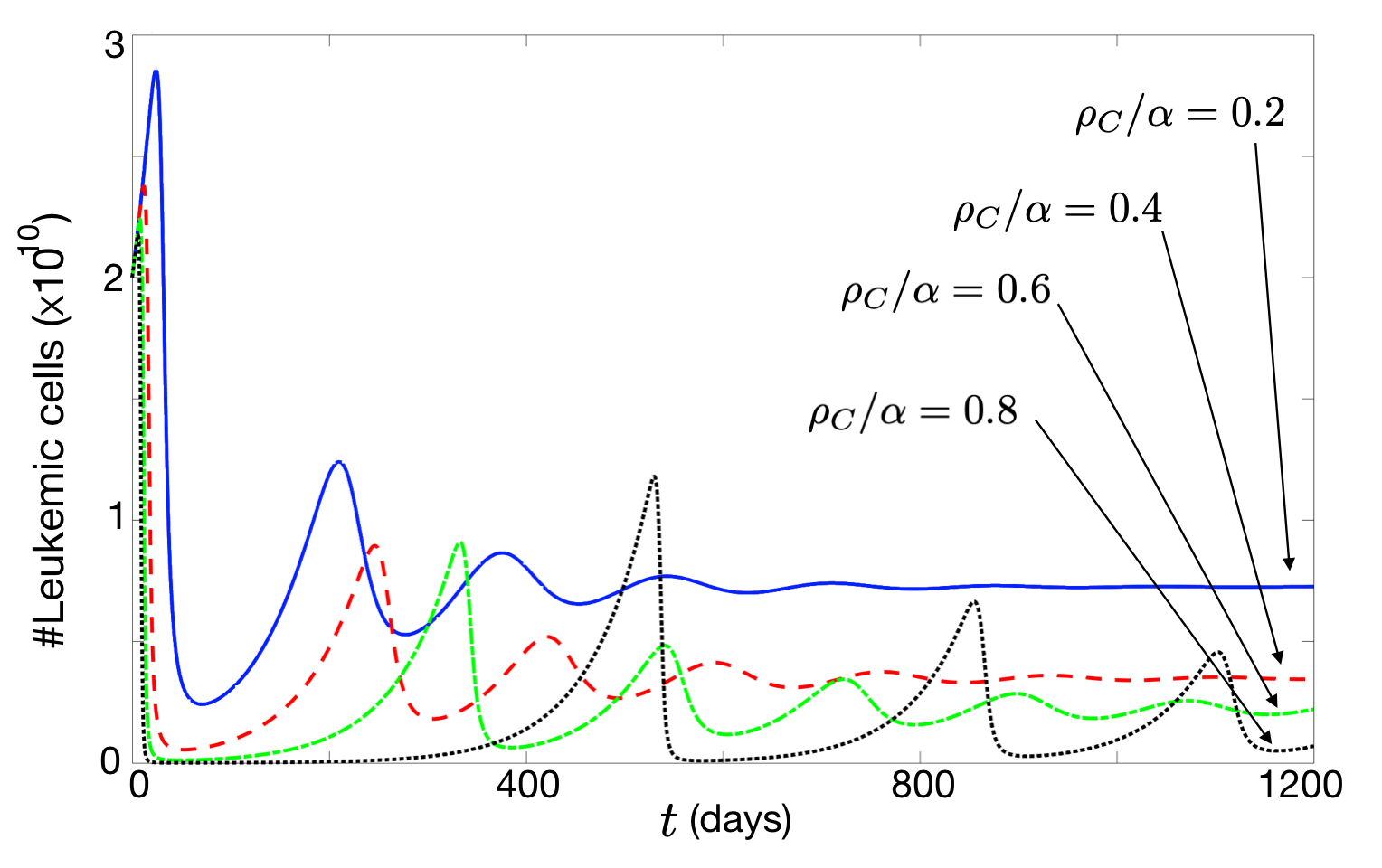}
	\caption{\textbf{Higher mitotic stimulation rates provide better tumor control.} Dynamics of 
	the leukemic population governed by Eq. (\ref{model1})  for initial data $T_0 = 10^9, L_0 = 2\times 10^9, C_0 = 10^5$ cells, and parameter values $\tau_C = 14$ days, $\rho_L = 1/60$ day$^{-1}$, $\alpha = 5.84 \times 10^{-11}$ day$^{-1}$ cell$^{-1}$. The different curves correspond to stimulation rates
	$\rho_C = 0.2 \alpha$ (blue solid line), $\rho_C = 0.4 \alpha$ (red dashed line), $\rho_C = 0.6 \alpha$ (green, dash-dotted line), $\rho_C= 0.8 \alpha$ (black dotted line).}
\label{equilibriaL}
\end{figure}

\subsection{Initial number of CAR-T cells injected does not affect the outcome of therapy}

We next studied the effect of the number of CAR-T cells initially injected on the system's dynamics. To do so, we performed an extensive number of simulations over the biologically feasible range and found a very weak dependence of the dynamics on the number of injected CAR-T cells. An example is shown in Fig. \ref{equilibriaX} for a broad range of cells initially injected ranging from $10^4$ to $10^6$. 
Although there was a difference of two orders of magnitude in $C_0$, it led to a small variation in the time to peak expansion of a few days, a negligible increase of the maximum CAR-T and leukemic cell number, and a minor differences in the times to tumor relapse. 

 \begin{figure}
\centering
	\includegraphics[width=0.9\textwidth]{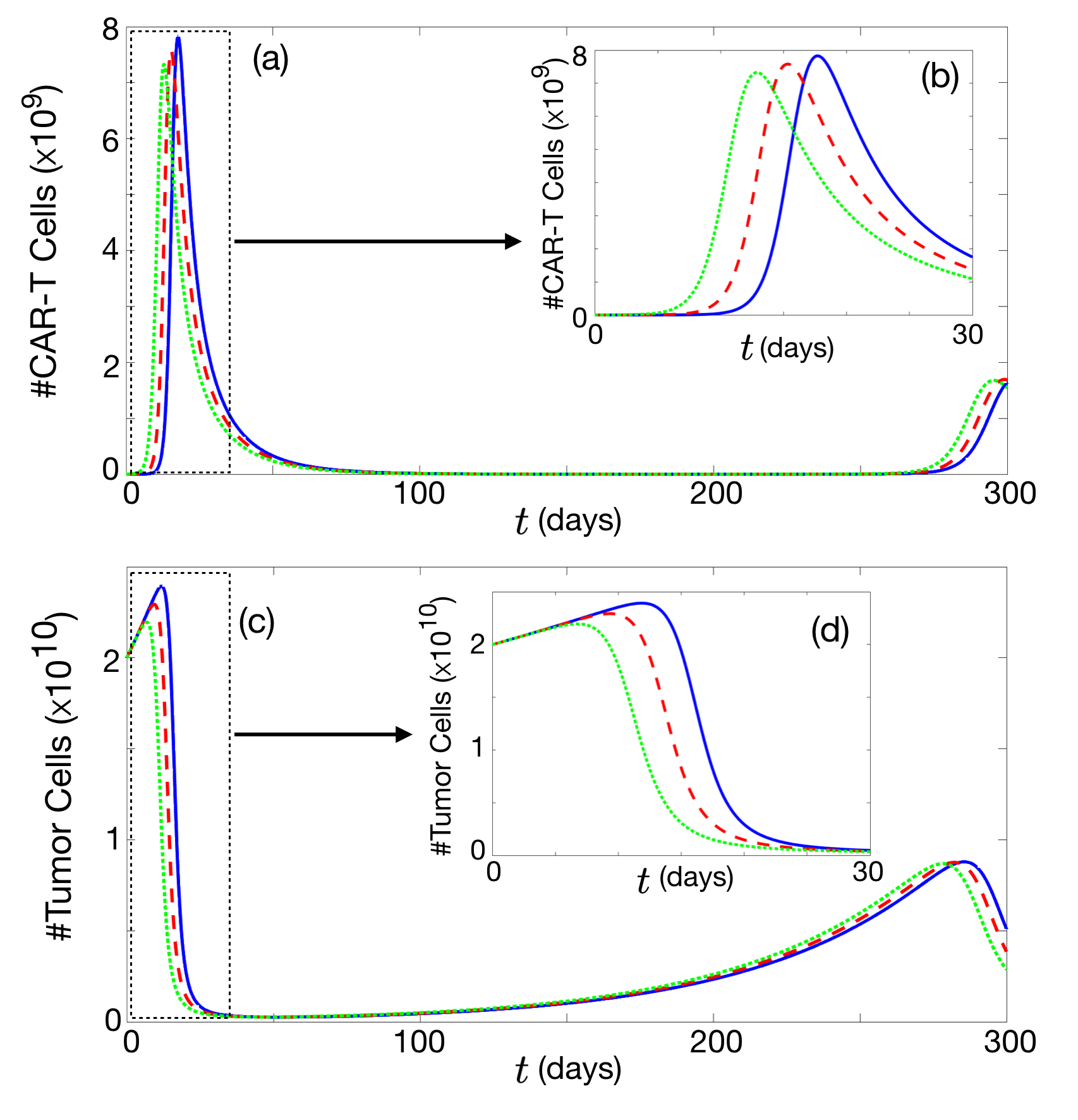}
	\caption{\textbf{Initial number of CAR-T cells injected does not affect the outcome of therapy.} Dynamics of the number of CAR-T cells (a,b) and leukemic cells (c,d) governed by Eq. (\ref{model1}) over the time range [0,300] days (a,c). We also show the details of the initial response of the treatment over the time interval [0,30] (b,d). Initial data used in the simulations were $T_0 = 10^9, L_0 = 2\times 10^9$ cells, and parameter values $\tau_C = 14$ days, $\rho_L = 1/60$ day$^{-1}$, $\alpha = 5.84 \times 10^{-11}$ day$^{-1}$ cell$^{-1}$. The curves correspond to different values of 
	$C_0 = 10^4$ (blue solid line), $C_0 = 10^5$ (red dashed line), $C_0 = 10^6$ (green dash-dotted line).} \label{equilibriaX}
\end{figure}

\subsection{Tumor proliferation rate did not essentially affect either the initial response or the asymptotic leukemic cell values, but did influence relapse time}

Finally, we studied the dynamics under modifications of the tumor proliferation rate in the whole feasible range for fast-growing leukemias $\rho_C \in [1/60,1/20]$. The result of typical simulations are shown in Figure \ref{rhoLL}.

The short term dynamics in response to the CAR-T injection were found to be qualitatively similar independently of the tumor proliferation rate, with a peak tumor cell number around day 10 post-injection (Figure \ref{rhoLL}(b)). Small differences were observed in the peak number of tumor cells obtained, essentially due to the fact that faster growing leukemic cells could grow further during the initial CAR-T expansion phase. After CAR-T cell expansion, there was a response phase for the different tumor growth speeds and then a relapse was observed (Figure \ref{rhoLL}(a)). The relapse time was found to depend on the proliferation rate. This is in line with what one would expect for the amplitude and frequency of oscillations towards the equilibrium point $E_2$, which are both proportional to $\rho_L$ according to Eq. (\ref{equifor}).

As expected from the expression for $E_2$, and the values of the parameters, there was a weak dependence of the number of leukemic cells
in the equilibrium on $\rho_L$  in the range of relevance (Figure \ref{rhoLL}(c)), given analytically by $L_2 = 1.5289 \times 10^9 + 4.2808 \times 10^9 \rho_L$, with $\rho_L \in [1/60,1/20]$ day$^{-1}$. Thus the major contribution to the asymptotic tumor cell count was $L_2 \sim 1/(\rho_C \tau_C)$. 

Although the therapy had a substantial effect, logarithmic scale plots (Figure \ref{rhoLL}d) show the persistence of measurable disease for all times. 

Let us define the maximum tumor cell load reduction achieved by the treatment  as 
\begin{equation}
R = \max_t \left(T(t)\right) / \min_t \left(T(t)\right).
\end{equation} 
In the simulations shown in Figure \ref{rhoLL}, 
this quantity was found to be $R(\rho_L = 1/20) = 65$, $R(\rho_L = 1/30) = 72$, $R(\rho_L = 1/40) = 78$, $R(\rho_L = 1/50) = 82$, $R(\rho_L = 1/60) = 88$, thus always smaller than 100 (two orders of magnitude). It is easy to see that CAR-T cells decreased in number over time, as did the tumor load, but they were always above the numbers of cells initially injected. In fact, for most tumor proliferation rates the number of CAR-T cells was more than one order of magnitude above the level injected.

\begin{figure}
\centering
	\includegraphics[width=0.95\textwidth]{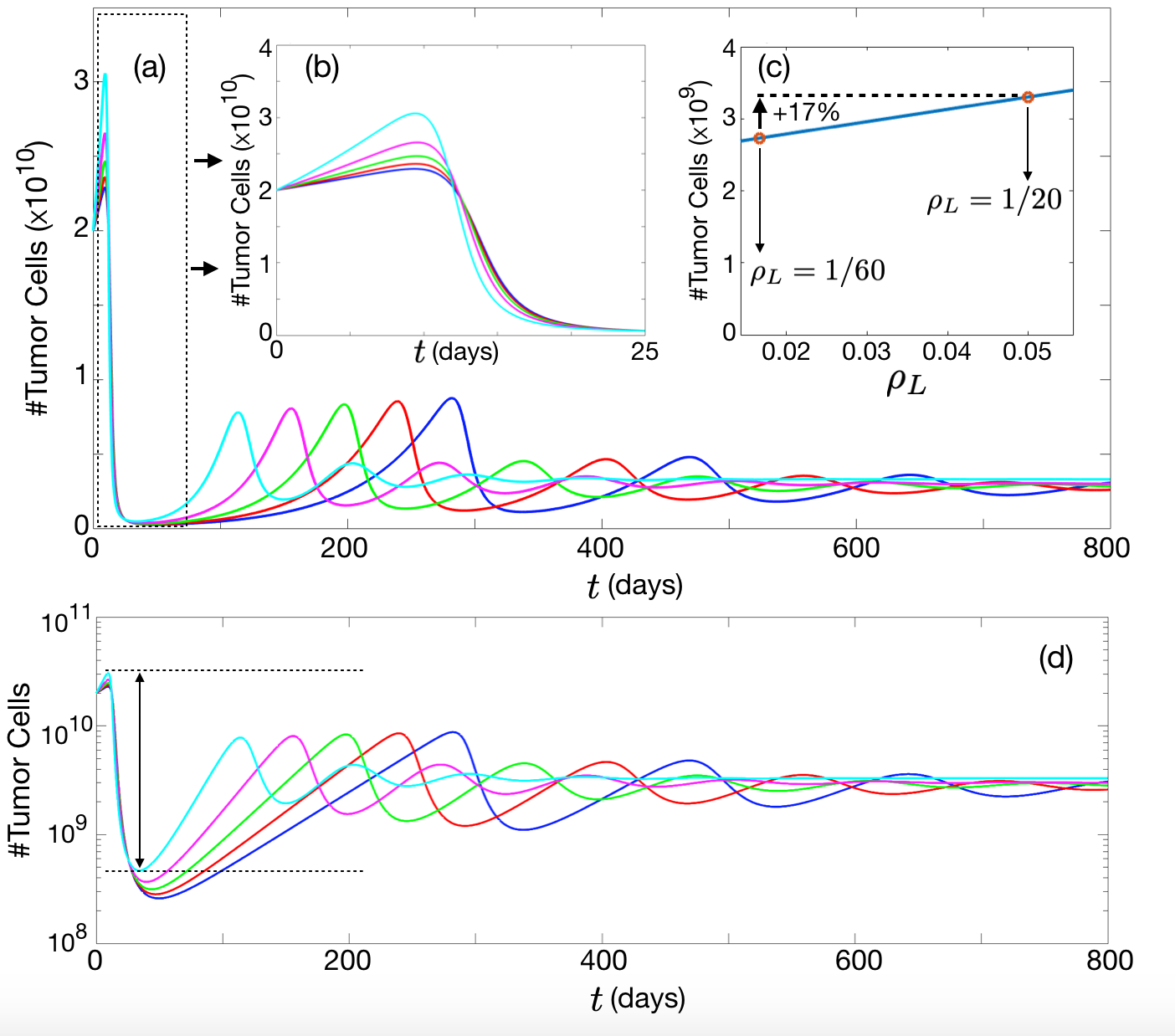}
	\includegraphics[width=0.95\textwidth]{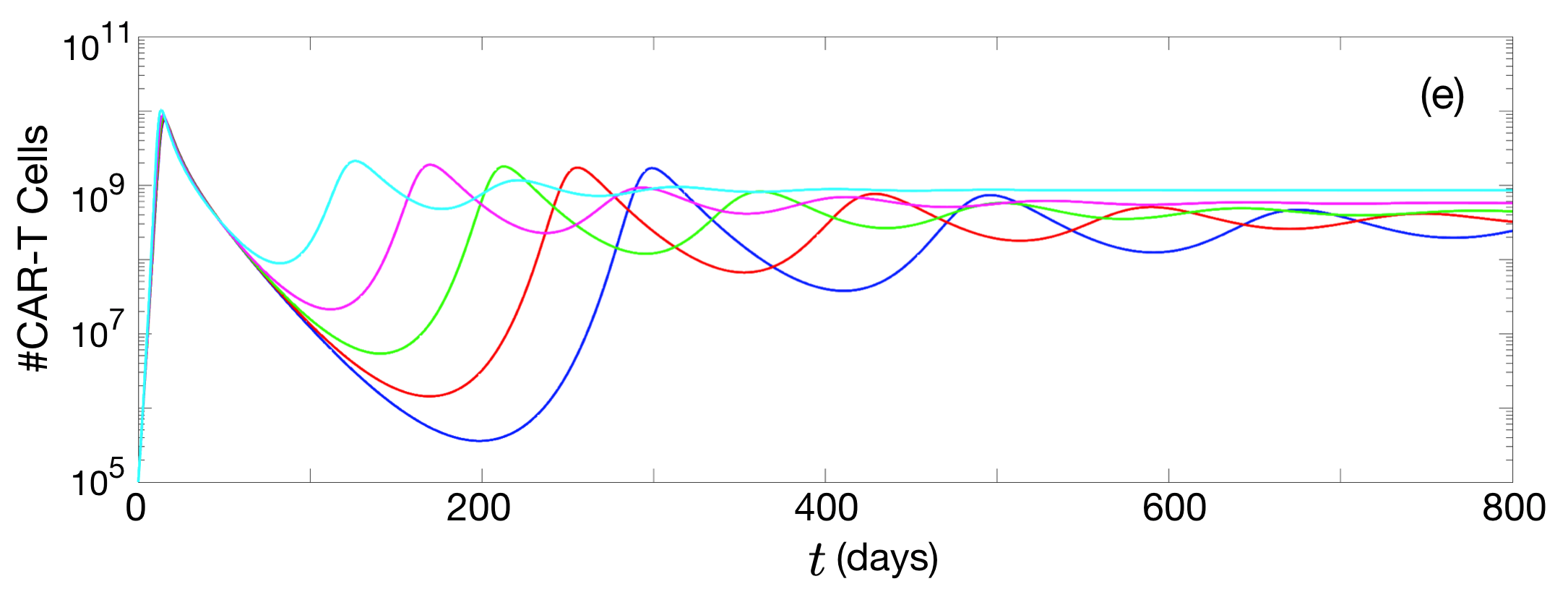}
	\caption{\textbf{Tumor proliferation rate did not affect either the initial response or the asymptotic values, but did influence relapse time.} Dynamics of 
	the leukemic population governed by Eq. (\ref{model1})  for initial data $T_0 = 10^9, L_0 = 2\times 10^9, C_0 = 10^5$ cells, and parameter values $\tau_C = 14$ days,  $\alpha = 5.84 \times 10^{-11}$ day$^{-1}$ cell$^{-1}$, $\rho_C=0.5 \alpha$, for different values of $\rho_L$. The different curves correspond to different values of 
	$\rho_L = 1/60$ (blue line), $\rho_L = 1/50$ (red line), $\rho_L = 1/40$ (green line), $\rho_L = 1/30$ (magenta line), $\rho_L = 1/20$ (cyan line) (a) Dynamics over the time range [0,800] days. (b) Details of the dynamics for the time interval [0,25]. (c) Dependence of the asymptotic tumor values obtained from Eq. (\ref{L2e}). (d) Tumor cell number evolution in logarithmic scale. The rate between maximum and minimum tumor load is indicated with an arrow for the case $\rho_L = 1/20$. (e) Evolution of the number of CAR-T cells.}
\label{rhoLL}
\end{figure}

\subsection{CAR-T cell reinjection does not improve the therapy outcome.}
 
 An interesting question is if one could control relapses by acting on the tumor by reinjecting CAR T cells. 
 To study that we performed extensive numerical simulations over the parameter range of interest. 
 Figure \ref{reinjected} summarizes some results where we 
 simulated the reinjection of $C =10^5$ CAR-T cells at different times post-injection for different tumor growth rates and quantified the variations in maximum tumor load with respect to the case without reinjection.  In the best scenario, corresponding to slowly growing tumors, the improvement in the peak tumor cell number at relapse was around 2\%. Thus, the CAR-T reinjection did not substantially improve the outcome for any delay nor tumor proliferation rate.

\begin{figure}
\centering
\includegraphics[width=0.95\textwidth]{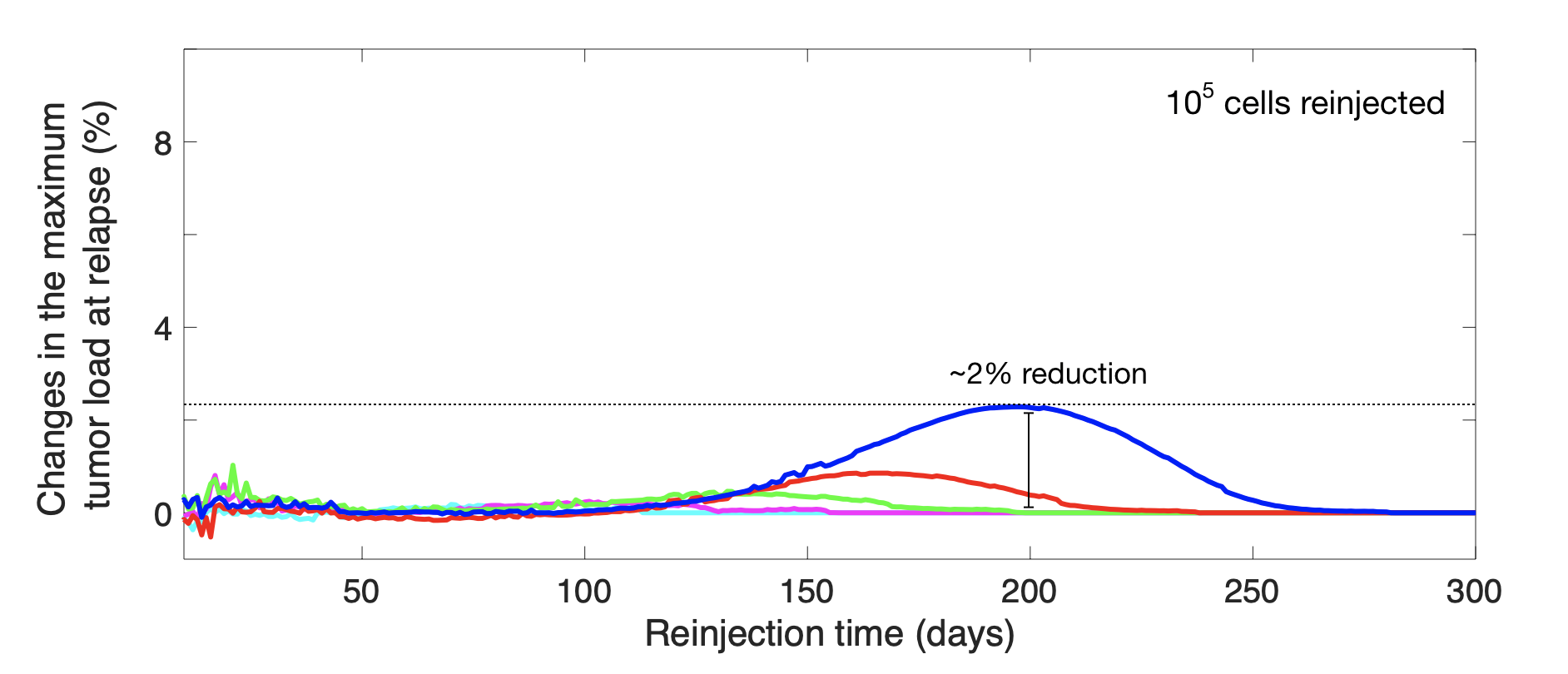}
\caption {\textbf{CAR-T cell reinjection does not improve the therapy outcome.} Variations of the maximum tumor load at the first relapse when reinjecting $C = 10^5$ CAR-T cells for different reinjection times using initial data $T_0 = 10^9, L_0 = 2\times 10^9, C_0 = 10^5$ cells, and parameter values $\tau_C = 14$ days,  $\alpha = 5.84 \times 10^{-11}$ day$^{-1}$ cell$^{-1}$, $\rho_C=0.5 \alpha$. The different curves correspond to different values of the tumor proliferation rate $\rho_L = 1/60$ (blue line), $\rho_L = 1/50$ (red line), $\rho_L = 1/40$ (green line), $\rho_L = 1/30$ (magenta line), $\rho_L = 1/20$ (cyan line). }
\label{reinjected}
\end{figure}

\section{Discussion and Conclusion}
\label{conclu}
 
CAR-T cell therapies for B-cell malignancies is one the most resounding successes of the current immunotherapies, driving a strong interest on the topic \cite{Cell}. Several mathematical models have been constructed describing the observed dynamics \cite{Sahoo,Baar,Kimmel,Rodrigues,Anna,Ode}. This has lead to very considerable interest in extending these therapies to other hematological tumors, such as malignant T-cell leukemias.

As stated in the introduction, one of the challenges faced by these treatments is fratricide, i.e. the fact that CAR-T cells, belonging to the T-cell lineage and expressing common antigens with the leukemic cells, would themselves become targets of the therapy. This poses the very interesting question of what the outcome of such a therapy would be, given that it poses challenges even for CAR-T cell production in vitro. Interestingly, our simple mathematical model captured the difficulties for CAR-T cell expansion in vitro, with a limit in cell production given by Eq. (\ref{expansion}). Thus, the maximum number of CAR-T cells that can be produced in vitro depends on the stimulation provided by the cytokines and the excess CAR-T killing efficiency over the mitotic stimulation.

One might naively think that in-vivo expansion would also be limited, not having a substantial effect on the disease. However, this is not true. Our simulations showed that when they are injected, even in the small numbers that can be obtained in vitro, the CAR-T cells find many targets initially on both the healthy and tumor T-cells. During this initial stage, the CAR-T population is amplified even in the presence of fratricide. We also found in silico that the outcome of the therapy did not depend on the number of CAR-T cells injected. 

A relapse was always observed in the framework of our model simulations and the number of tumor cells was initially reduced by a factor smaller than 100, with the persistence of measurable disease for all times, so the treatment did not eradicate the disease in our numerical simulations. However, CAR-T cells were able to control tumor growth after two weeks and then, even in spite of the oscillations, the high initial tumor loads were never found to appear again..
 
Relapse time after the CAR-T treatment was found to depend strongly on the proliferation rate. This makes us wonder if a combination with a post-CAR chemotherapy could prove useful in delaying tumor regrowth. The opposite strategy, i.e. first giving chemotherapy and then CAR-T before relapse, would not be recommended, however, since chemotherapy would be expected to reduce the number of target cells and then CAR-T would lack the substrate to expand. 
 
 Another question related to these treatments is whether re-challenging with CAR-T cells at any given time could be beneficial. For instance one may wonder if that process could be used to delay or even eliminate tumor relapse. On the basis of our computational results, additional injections were found to have no substantial effect on the dynamics. The reason is that CAR-T cells decreased in number transiently in time  after reaching peak expansion, but their levels were always above the numbers of cells initially injected, typically by more than an order of magnitude. This means that injecting small numbers of CAR-T cells in relation to those already present would not have a substantial effect on the dynamics. 
 
 The fact that the equilibrium $L_2 \sim 1/(\rho_C \tau_C)$ implies that there are two ways to improve the long-term efficacy of CAR-T cell therapy for T-cell leukemias. The first would be to improve the persistence of the CAR-T cells, something that has been done for B-cell leukemias by using CD19 CAR (CAT) with lower affinity than FMC63, the high-affinity binder used in many clinical studies \cite{Ghorashian}. The second would be to improve the mitotic stimulation rate, but keeping in mind the restriction $\alpha > \rho_C$. 
 
 In conclusion, in this paper we have developed a mathematical model of the dynamics of leukemic cells, healthy T-cells and CAR-T cells, after the therapeutic injection of the latter population. The mathematical model showed the potential of the treatment to control, but not eradicate, the disease. This would result in a chronification of the disease that could last for a long time, or it could buy some time to try alternative therapeutic strategies. Our work is a first simple mathematical attempt to cast light on the potential outcomes of these treatments. There are different types of T-cell malignancies and specifically T-cell leukemias, and the particular features of each type could be incorporated into more detailed models including additional biological details. We hope our work will stimulate further work in this exciting sub-field of immunotherapy.
 
\section*{Acknowledgements}

This work has been partially supported by the Junta de Comunidades de Castilla-La Mancha (grant number SBPLY/17/180501/000154), the James S. Mc. Donnell Foundation (USA) 21st Century Science Initiative in Mathematical and Complex Systems Approaches for Brain Cancer (Collaborative award 220020450), Junta de Andalucía group FQM-201, Fundación Española para la Ciencia y la Tecnología (FECYT, project PR214 from the University of Cádiz) and the Asociación Pablo Ugarte (APU). OLT is supported by a PhD Fellowship from the University of Castilla-La Mancha research plan.

We would like to acknowledge Gabriel F. Calvo, Carmen Ortega-Sabater, Juan Belmonte Beitia (MOLAB, University of Castilla-La Mancha, Spain), Manuel Ram\'{\i}rez-Orellana (Hospital Universitario Ni\~no Jes\'us, Madrid, Spain) and  Soukaina Sabir (University Mohamed V, Morocco) for discussions.

   \end{document}